\documentclass[aps,pra,10pt,a4paper,reprint,nofootinbib,superscriptaddress]{revtex4-2}
\usepackage[utf8]{inputenc}
\usepackage{amsfonts}
\usepackage{amssymb}
\usepackage{amsmath}
\usepackage{amsthm}
\usepackage{graphics}
\usepackage{graphicx}
\usepackage{braket}
\usepackage{color}
\usepackage[colorlinks=true,linkcolor=blue,urlcolor=blue,citecolor=blue]{hyperref}
\usepackage{times,txfonts}
\newtheorem{theorem}{Theorem}

\newtheorem{proposition}{Proposition}
\newtheorem{remark}{Remark}

\theoremstyle{definition}

\newcommand{\ketbra}[2]{|#1\rangle \langle #2|}
\definecolor{green(munsell)}{rgb}{0.0, 0.7, 0.47}

\begin{document}
\title{Entanglement of minimal dimension and a class of local state discrimination problems}

\author{Saronath Halder}
\email{saronath.halder@vitap.ac.in}
\affiliation{Department of Physics, School of Advanced Sciences, VIT-AP University, Beside AP Secretariat, Amaravati 522241, Andhra Pradesh, India}

\author{Suchetana Goswami}
\email{suchetana.goswami@gmail.com}
\affiliation{School of Informatics, University of Edinburgh, 10 Crichton Street, Newington, Edinburgh EH8 9AB, United Kingdom}

\begin{abstract}
In this work, we construct small sets of bipartite orthogonal pure states that cannot be perfectly distinguished by local operations and classical communication (LOCC). We mention that not all the states within the constructed sets are necessarily entangled. However, such a set contains at least one entangled state which cannot be conclusively identified by LOCC (with nonzero probability). Then, we show that the states of any such set can be perfectly distinguished by LOCC using a minimal-dimensional entangled resource state. Clearly, here the entangled resource state provides an advantage irrespective of the dimension of the given set. Using this result, we also prove that any pure entangled state is useful as a resource to distinguish the states of any present set unambiguously with nonzero probability under LOCC. These sets exist in all two-qudit Hilbert spaces. Furthermore, it is possible to decrease the average entanglement content of such a set or to increase the cardinality of the set without changing certain properties of the same. 
\end{abstract}
\maketitle

\section{Introduction}\label{sec1}
Quantum state discrimination \cite{Chefles00, Barnett09, Bae15} is a well-known problem in quantum information science. In this problem, a quantum system is prepared in a state which is secretly taken from a known set, the goal is to identify that state. If the states of the given set are pairwise orthogonal, then, in principle, in which state the system is prepared can always be identified perfectly. However, even if we deal with orthogonal states, the state identification may not be possible when the parts of a composite system are distributed among several spatially separated parties and the parties of those locations are allowed to perform only local operations and classical communication (LOCC) \cite{Chitambar14-1}. This is the version of the state discrimination problem that we consider here. So, if we rephrase the problem for our purpose, then it can be stated in the following manner. We consider a composite quantum system that is prepared in a state and the state is secretly taken from a given set. Now, the system is distributed between two spatially separated parties. Then, the parties of those locations perform LOCC to identify the state of the system. In our case, the given set contains only orthogonal pure states. Additionally, the states within a given set are equally probable.

There are many examples of orthogonal states that cannot be perfectly distinguished by LOCC. For example, for product states, see \cite{Bennett99-1, Bennett99, Walgate00, DiVincenzo03, Horodecki03, Niset06, Feng09, Zhang14, Chitambar14, Zhang15, Xu16, Halder18, Halder19} and for entangled states, see \cite{Ghosh01, Ghosh02, Ghosh04, Fan04, Nathanson05, Watrous05, Bandyopadhyay10, Bandyopadhyay11, Yu12, Yang13, Cosentino14, Banik21}. We mention that the state discrimination problem under LOCC (also known as local state discrimination problem) is very important direction of research as it finds application in data hiding \cite{Terhal01, Eggeling02, DiVincenzo02, Lami18, Lami21, Ha25} and in secret sharing \cite{Rahaman15}. When for a given set, the state of the system cannot learned optimally under LOCC, we say that the set is locally indistinguishable or the states are locally indistinguishable. Otherwise, the set (or the states) is (are) locally distinguishable. Usually, when we deal with orthogonal states, we seek for a strategy to distinguish the states perfectly, i.e., to identify the state of the system with certainty. Now, when perfect discrimination is not possible, one can still think about probabilistic discrimination. There are two usual approaches in this direction, (i) discrimination by minimizing the error, (ii) discrimination without committing any error. The former is called minimum-error state discrimination and the later is called the conclusive or the unambiguous state discrimination. For our purpose, the second one is required. In fact, in the second case, there is no scope of committing any error but there is a possibility of an inconclusive outcome.

In this work, we study the sets that are locally indistinguishable but entanglement of minimal dimension as a resource might be sufficient to make them distinguishable under LOCC. For entanglement-assisted local state discrimination processes, see \cite{Cohen08, Bandyopadhyay09, Bandyopadhyay16, Zhang18, Halder18, Rout19, Shi20, Zhang20, Yu14, Gungor16, Bandyopadhyay18, Lovitz22}. The motivation of studying the present sets follows from the fact that in an information processing protocol, we encode the information in quantum states and then, when it is required, there must be a way to decode the information using additional resources. Now, to make the process efficient, one can think about using a minimal-dimensional resource state in the information decoding. For example, one can have a look into the Refs.~\cite{Goswami23, Srivastav25}.

Here, we construct incomplete sets, i.e., the sets that do not form bases for the considered Hilbert spaces. Among such incomplete sets, there are some sets which contain only entangled states and exhibit local indistinguishability. For example, it is known that in a two-qudit system, any set that contains $d+1$ orthogonal maximally entangled states cannot be perfectly distinguished by LOCC ($d$ is the dimension of the individual subsystems) \cite{Ghosh04, Hayashi06}. However, we prefer not to take all states as entangled. Furthermore, if possible, we also do not want to stick to maximal entanglement for the entangled states. Additionally, if we consider less than $d+1$ cardinality then, it is not always possible to get local indistinguishability \cite{Walgate00, Nathanson05}. However, our goal is to put forward a construction which works for all two-qudit Hilbert spaces and also the average entanglement content of the set is less. 

The present contributions can be given as the following. We first construct a small set of $d+1$ states in a two-qudit system; $d$ is the dimension of the smaller subsystem. Such a set contains only orthogonal pure states, $d$ of them can be maximally entangled and one is a product state. We prove that for such a set not all states are conclusively locally identifiable. Therefore, obviously, the states of the set cannot be perfectly distinguished by LOCC. Furthermore, for local indistinguishability, it is not necessary that the entangled states must be maximally entangled. Then, we show that a two-qubit maximally entangled state as resource is sufficient to distinguish the states of any present set irrespective of the values of the dimensions of subsystems. This clearly depicts the power of a two-qubit entangled state though it is a minimum-dimensional entangled state. Using this result, we also prove that all pure entangled states are useful in distinguishing the states of any present set conclusively under LOCC. Additionally, it is shown that in some cases it is possible to increase the cardinality of a set without changing its property in the context of entanglement-assisted discrimination.

The rest of the paper is arranged in the following manner. In Sec.~\ref{sec2}, we present some preliminary concepts for understanding the main results. Then, in Sec.~\ref{sec3}, we present the results step by step. Finally, we draw the conclusion in Sec.~\ref{sec4}, mentioning some open problems.

\section{Preliminaries}\label{sec2}
We focus on a bipartite system, associated with a Hilbert space $\mathcal{H} = \mathbb{C}^d\otimes\mathbb{C}^{d^\prime}$. More precisely, we write the dimension of the bipartite system as, $d\otimes {d^\prime}$; where $d$ is the dimension of the smaller subsystem, i.e., $d^{\prime}\geq d$. If we consider a state of maximal Schmidt rank, then it can have maximum $d$ Schmidt rank in $d\otimes d^{\prime}$. Clearly, when we consider a set containing a maximal Schmidt rank state, the set belongs to both $d\otimes d$ as well as $d\otimes d^{\prime}$. For this reason, in the result section, we discuss about sets in $d\otimes d$ only ($d\geq2$). Because, for all values of $d$, if we construct sets in $d\otimes d$, then, it implies that this type of sets exist in all bipartite Hilbert spaces. We mention that the minimum number of product states required to express an entangled state is equal to the Schmidt rank of the entangled state \cite{Horodecki09}. 

When we introduce a set and discuss about the local discrimination of the states of the given set, we may not consider normalized states because normalization may not play any important role there. In $d\otimes d$, the set we consider has a cardinality $d+1$. Among $d+1$ states, $d$ states are maximally entangled states and one state is a product state. Later, we also provide sets where non-maximally entangled states are used, there providing normalization becomes important.

For perfect discrimination of states, it is necessary that the states of the given set are conclusively distinguishable. Furthermore, for conclusive distinguishability of the states of a given set, it is required to check if each state of the set is conclusively identifiable. In the following, we provide a necessary and sufficient criterion to check the conclusive identification of a state from a given set under LOCC \cite{Chefles04, Bandyopadhyay09-1}. We consider that a set of bipartite states $\{\ket{\phi_i}\}_{i=1}^n$ is given. Now, the state $\ket{\phi_i}$ can be identified conclusively if there exist a state of the form $\ket{\alpha}\ket{\beta}\equiv\ket{\alpha\beta}$, such that $\braket{\alpha\beta|\phi_i}>0$ and $\braket{\alpha\beta|\phi_j}=0$, where $j \in \{1,2,\dots, n\}$ and $j\neq i$.

We are now ready to discuss our main findings in the following section.

\section{Results}\label{sec3}
We start with a two-qubit system. In $2\otimes2$, we consider a set consisting of the following states:
\begin{equation}\label{eq1}
\ket{\psi_1} = \ket{00} + \ket{11},~ 
\ket{\psi_2} = \ket{00} - \ket{11},~ 
\ket{\psi_3} = \ket{01}.
\end{equation}
It is known that if a set contains two entangled states and a product state in a two-qubit system, then these states cannot be perfectly distinguished by LOCC \cite{Walgate02}. For this reason, the above set is locally indistinguishable. Now, since the set is a two-qubit set, a maximally entangled state is sufficient to distinguish the states via a teleportation-based protocol. 

Next, we consider a similar set for a two-qutrit system, and the states in the set are given by,
\begin{equation}\label{eq2}
\begin{array}{l}
\ket{\psi_1} = \ket{00} + \ket{11} + \ket{22},\\[0.5 ex]

\ket{\psi_2} = \ket{00} +\omega\ket{11} + \omega^2\ket{22},\\[0.5 ex]

\ket{\psi_3} = \ket{00} +\omega^2\ket{11} + \omega\ket{22},\\[0.5 ex]

\ket{\psi_4} = \ket{01},
\end{array}
\end{equation}
where $\omega$ and $\omega^2$ are cube roots of unity such that $\omega^3 = 1$ and $1+\omega+\omega^2=0$. It is known that the states $\{\ket{\psi_1},~\ket{\psi_2},~\ket{\psi_4}\}$ of the above set cannot be perfectly distinguished by LOCC. These states are discussed in the context of the phenomenon ``more nonlocality with less entanglement'' \cite{Horodecki03}. For this reason, the above set is also locally indistinguishable. However, we want to discuss the entanglement-assisted discrimination of the set given in (\ref{eq2}). In this context, we prove the following.

\begin{proposition}\label{prop1}
A two-qubit maximally entangled state is sufficient to distinguish the states of Eq. (\ref{eq2}) perfectly via LOCC.   
\end{proposition}

\begin{proof}
To prove the above, we provide an explicit protocol for state discrimination. We consider the two-qubit entangled state as $\ket{\Psi} = \ket{00} + \ket{11}$. The set along with the resource state $\{\ket{\psi_1}$, $\ket{\psi_2}$, $\ket{\psi_3}$, $\ket{\psi_4}\}\otimes\ket{\Psi}$ is given as the following, 

\begin{widetext}
\begin{equation}\label{eq3}
\begin{array}{l}
\ket{00}\ket{00} + \ket{10}\ket{10} + \ket{20}\ket{20} + \ket{01}\ket{01} + \ket{11}\ket{11} + \ket{21}\ket{21}, \\[0.5 ex]

\ket{00}\ket{00} + \omega\ket{10}\ket{10} + \omega^2\ket{20}\ket{20} + \ket{01}\ket{01} + \omega\ket{11}\ket{11} + \omega^2\ket{21}\ket{21}, \\[0.5 ex]

\ket{00}\ket{00} + \omega^2\ket{10}\ket{10} + \omega\ket{20}\ket{20} + \ket{01}\ket{01} + \omega^2\ket{11}\ket{11} + \omega\ket{21}\ket{21}, \\[0.5 ex]

\ket{00}\ket{10} + \ket{01}\ket{11}.
\end{array}
\end{equation}  
\end{widetext}
The protocol proceeds the following way. First, Alice performs a two-outcome projective measurement defined by two projection operators. They are given by, \begin{equation}\label{eq4}
\begin{array}{l}
\Pi_1 = \ketbra{00}{00}+\ketbra{11}{11}+ \ketbra{21}{21},\\[0.5 ex]

\Pi_2 = \ketbra{01}{01}+\ketbra{10}{10}+ \ketbra{20}{20},
\end{array}
\end{equation}
with outcomes ``1" and ``2" respectively. If the outcome is ``1'', then Alice and Bob are left with the following states.
\begin{equation}\label{eq5}
\begin{array}{l}
\ket{00}\ket{00} + \ket{11}\ket{11} + \ket{21}\ket{21}, \\[0.5 ex]

\ket{00}\ket{00} + \omega\ket{11}\ket{11} + \omega^2\ket{21}\ket{21}, \\[0.5 ex]

\ket{00}\ket{00} + \omega^2\ket{11}\ket{11} + \omega\ket{21}\ket{21}, \\[0.5 ex]

\ket{00}\ket{10}.
\end{array}
\end{equation}
Next, Bob does the same measurement defined by the operators $\{\Pi_1,~\Pi_2\}$. If the outcome is ``2'' then the state is $\ket{\psi_4}$. Otherwise, Alice and Bob are left with the following states.
\begin{equation}\label{eq6}
\begin{array}{l}
\ket{\mathbf{0}}\ket{\mathbf{0}} + \ket{\mathbf{1}}\ket{\mathbf{1}} + \ket{\mathbf{2}}\ket{\mathbf{2}}, \\[0.5 ex]

\ket{\mathbf{0}}\ket{\mathbf{0}} + \omega\ket{\mathbf{1}}\ket{\mathbf{1}} + \omega^2\ket{\mathbf{2}}\ket{\mathbf{2}}, \\[0.5 ex]

\ket{\mathbf{0}}\ket{\mathbf{0}} + \omega^2\ket{\mathbf{1}}\ket{\mathbf{1}} + \omega\ket{\mathbf{2}}\ket{\mathbf{2}},
\end{array}
\end{equation}
where $\ket{\mathbf{0}} \equiv \ket{00}$, $\ket{\mathbf{1}} \equiv \ket{11}$, and $\ket{\mathbf{2}} \equiv \ket{21}$. Interestingly, this set can be perfectly distinguished by LOCC when Alice and Bob both perform projective measurements in the following basis: 
\begin{equation}\label{eq7}
\begin{array}{l}
\ket{\mathbf{0}} + \ket{\mathbf{1}} + \ket{\mathbf{2}}, \\[0.5 ex]

\ket{\mathbf{0}} + \omega\ket{\mathbf{1}} + \omega^2\ket{\mathbf{2}}, \\[0.5 ex]

\ket{\mathbf{0}} + \omega^2\ket{\mathbf{1}} + \omega\ket{\mathbf{2}}.
\end{array}
\end{equation}
Finally, we trace back the other outcome of the very first measurement by Alice. In that case too, similar steps will be followed. In this way, the whole protocol can be completed and it proves the above proposition.
\end{proof}

Note that, in the above protocol, Alice and Bob are performing disjoint measurements. Therefore, nowhere choosing a measurement among several measurements based on any type of classical communication is involved. Clearly, when an additional entangled state is provided as a resource, the given states can be distinguished instantaneously, i.e., under LO only, instead of using full LOCC. This is an important feature that becomes useful in several cryptographic protocols, such as entity authentication, where CC can reveal the underlying secret; yet just with LO the shared quantum state can be identified \cite{goswami25}.

Now we also discuss about increasing the cardinality, i.e., one can think about including other states within the set of Eq. (\ref{eq2}). For example, if we consider another product state $\ket{02}$ with the states of Eq. (\ref{eq2}), then, along with the resource state, the new product state becomes $\ket{00}\ket{20}+\ket{01}\ket{21}$. Once Alice performs her measurement and gets outcome ``1'', the state becomes $\ket{00}\ket{20}$. Clearly, if Bob gets outcome ``2'' after his measurement, the parties are left with two orthogonal product states which can be distinguished by some local measurements for sure \cite{Walgate00}. The rest of the protocol goes in the same way as it is described in the proof of the preceding proposition. In this way, the cardinality can be increased without altering the result related to the entanglement-assisted discrimination given in the same proposition (Proposition \ref{prop1}). Note that, if cardinality is increased then, full LOCC might be required in case of entanglement-assisted discrimination.\\

\begin{remark}\label{rem1}
From the above proposition, it is quite clear that the subset, i.e., $\{\ket{\psi_1},~\ket{\psi_2},~\ket{\psi_4}\}$ of Eq. (\ref{eq2}) which arises in the context of ``more nonlocality with less entanglement'' can be perfectly distinguished by LOCC if a two-qubit maximally entangled state is available as a resource.
\end{remark}

Now we proceed to discuss about a set in $4\otimes4$. The set is given as the following. 
\begin{equation}\label{eq8}
\begin{array}{l}
\ket{\mu_1} = \ket{00} + \ket{11} + \ket{22} + \ket{33},\\

\ket{\mu_2} = \ket{00} + \ket{11} - \ket{22} - \ket{33},\\[0.5 ex]

\ket{\mu_3} = \ket{00} - \ket{11} - \ket{22} + \ket{33},\\[0.5 ex]

\ket{\mu_4} = \ket{00} - \ket{11} + \ket{22} - \ket{33},\\[0.5 ex]

\ket{\mu_5} = \ket{01}.
\end{array}
\end{equation}
Then, regarding this set, we prove the following.

\begin{proposition}\label{prop2}
The above set in $4\otimes 4$ contains at least one state that cannot be conclusively identified by LOCC with nonzero probability. Clearly, the states cannot be perfectly distinguished by LOCC. However, these states can be perfectly distinguished by LOCC if a two-qubit maximally entangled state as a resource is available.    
\end{proposition}

\begin{proof}
The proof can be presented in two parts. The first part is for local indistinguishability and the second part is for entanglement-assisted discrimination. In the following, we first provide the proof of local indistinguishability. The technique is quite efficient in the sense that it can be easily generalized to higher dimensions.

We first recall that for perfect local distinguishability, it is necessary to distinguish the states probabilistically without committing any error. Clearly, if a set contains at least one state that cannot be conclusively locally identified, then such a set cannot be perfectly distinguished by LOCC. We further mention that the necessary and sufficient condition for conclusive local identification (also known as unambiguous local identification) \cite{Chefles04, Bandyopadhyay09-1} is described in Sec.~\ref{sec2}. 

Now, let us pick the state $\ket{\mu_1}$ 
and we search for the suitable product state. For conclusive local identification, this product state must have nonzero overlap with $\ket{\mu_1}$ and it should have zero overlap with the other states of the set. Such a product state must be written as superposition of the following states: $\{\ket{\mu_1}$, $\ket{02}$, $\ket{03}$, $\ket{10}$, $\ket{12}$, $\ket{13}$, $\ket{20}$, $\ket{21}$, $\ket{23}$, $\ket{30}$, $\ket{31}$, $\ket{32}\}$. Then, we consider the following:
\begin{equation}\label{eq9}
\begin{array}{l}
a_1(\ket{00} + \ket{11} + \ket{22} + \ket{33}) + a_2\ket{02} + a_3\ket{03} + a_4\ket{10} \\[0.5 ex]+ a_5\ket{12} + a_6\ket{13} + a_7\ket{20} + a_8\ket{21} + a_9\ket{23}\\[0.5 ex] + a_{10}\ket{30} + a_{11}\ket{31} + a_{12}\ket{32}  \\[1.5 ex]

=\ket{0}(a_1\ket{0}+a_2\ket{2}+a_3\ket{3}) + \ket{1}(a_4\ket{0}+a_1\ket{1}+a_5\ket{2} \\[0.5 ex] +a_6\ket{3}) + \ket{2} (a_7\ket{0}+a_8\ket{1}+a_1\ket{2}+a_9\ket{3})\\[0.5 ex] + \ket{3}(a_{10}\ket{0}+a_{11}\ket{1}+a_{12}\ket{2}+a_1\ket{3}),
\end{array}
\end{equation}    
where $a_i$ are complex coefficients such that after taking superposition, it must produce a valid quantum state. Note that, the states $a_1\ket{0}+a_2\ket{2}+a_3\ket{3}$ and $a_4\ket{0}+a_1\ket{1}+a_5\ket{2}+a_6\ket{3}$ are always linearly independent as $a_1\neq0$. Clearly, this superposition will not be able to produce the product state we are looking for. In this way, one can argue that the state $\ket{\mu_1}$ cannot be conclusively locally identified. Therefore, the set cannot be perfectly distinguished by LOCC.

We are now left with the second part. This is same as the previous proof. The first two-outcome projective measurement can be revised by the projectors,
\begin{equation}\label{eq10}
\begin{array}{c}
\Pi_1 = \ketbra{00}{00}+\ketbra{11}{11}+ \ketbra{21}{21} + \ketbra{31}{31},\\[0.5 ex]

\Pi_2 = \ketbra{01}{01}+\ketbra{10}{10}+ \ketbra{20}{20} + \ketbra{30}{30}.
\end{array}
\end{equation}
Such a measurement is performed by both the parties and in this process the product state can either be identified or be eliminated. Then, the parties are left with four entangled states which can be distinguished locally if both Alice and Bob perform the measurement in the following basis, 
\begin{equation}\label{eq11}
\begin{array}{l}
\ket{\mathbf{0}} + \ket{\mathbf{1}} + \ket{\mathbf{2}} + \ket{\mathbf{3}}, \\[0.5 ex]

\ket{\mathbf{0}} + \ket{\mathbf{1}} - \ket{\mathbf{2}} - \ket{\mathbf{3}}, \\[0.5 ex]

\ket{\mathbf{0}} - \ket{\mathbf{1}} - \ket{\mathbf{2}} + \ket{\mathbf{3}}, \\[0.5 ex]

\ket{\mathbf{0}} - \ket{\mathbf{1}} + \ket{\mathbf{2}} - \ket{\mathbf{3}},
\end{array}
\end{equation}
(with,$\ket{\mathbf{0}} \equiv \ket{00}$, $\ket{\mathbf{1}} \equiv \ket{11}$, $\ket{\mathbf{2}} \equiv \ket{21}$, and $\ket{\mathbf{3}} \equiv \ket{31}$) when Alice and Bob both get outcome ``1'' after the measurement defined by Eq. (\ref{eq10}). Otherwise, $\ket{\mathbf{0}} \equiv \ket{01}$, $\ket{\mathbf{1}} \equiv \ket{10}$, $\ket{\mathbf{2}} \equiv \ket{20}$, and $\ket{\mathbf{3}} \equiv \ket{30}$ when Alice and Bob both get outcome ``2'' after the measurement defined by Eq. (\ref{eq10}). In this way, the protocol for entanglement-assisted discrimination can be completed.
\end{proof}

We are now ready to present the general construction, i.e., the two-qudit construction ($d\geq2$). For the formation of a similar set of cardinality $d+1$, we first construct the maximally entangled states. For this purpose, we consider the state $\ket{\nu_1} = \ket{00} + \ket{11} + \cdots + \ket{d-1~~ d-1}$. Then, for other entangled states we consider $\ket{\nu_k} = \sum_{j=0}^{d-1} \left(e^{\frac{2\pi i}{d}}\right)^{k^\prime j} \ket{jj}$; where $k = 2,3,\dots, d$, $k^\prime = k-1$ and $i=\sqrt{-1}$. In this way we construct the maximally entangled states of the set, $\{\ket{\nu_1},~\ket{\nu_2},\dots,\ket{\nu_d}\}$. Next, the product state can be given as, $\ket{\nu_{d+1}} = \ket{01}$. In this way, we construct a set of cardinality $d+1$. In the following, we provide the main theorem of this work.

\begin{theorem}\label{theo1}
The set $\{\ket{\nu_1},~\ket{\nu_2},\dots,\ket{\nu_{d+1}}\}$ ($d\geq2$) contains at least one state that cannot be conclusively identified by LOCC with nonzero probability. Clearly, the states cannot be perfectly distinguished by LOCC. Furthermore, a two-qubit maximally entangled state as a resource is sufficient for the perfect discrimination of the states by LOCC.
\end{theorem}

\begin{proof}
We first talk about the proof of local indistinguishability. We use the same proof technique as given in the above and it is possible to show that the state $\ket{\nu_1}$ cannot be conclusively identified by LOCC. The reason is the same. If we consider the product states from the rest of the dimension of the composite Hilbert space then, taking superposition of the state $\ket{\nu_1}$ and those product states, it is not possible to produce another product state. In this way, we can prove that the state $\ket{\nu_1}$ cannot be conclusively identified by LOCC. Thus, the set cannot be perfectly distinguished by LOCC as well. 

Then, we talk about the entanglement-assisted discrimination. To generalize the protocol, we first generalize the following measurement which is a two-outcome projective measurement.
\begin{equation}\label{eq12}
\begin{array}{l}
\Pi_1 = \ketbra{00}{00}+\ketbra{11}{11}+\cdots+ \ketbra{d-1~~1}{d-1~~1},\\[0.5 ex]

\Pi_2 = \ketbra{01}{01}+\ketbra{10}{10}+\cdots+\ketbra{d-1~~0}{d-1~~0}.
\end{array}
\end{equation}
Such a measurement is performed by both the parties and in this process the product state can either be identified or eliminated. The parties are then left with $d$ entangled state which can be distinguished locally if both Alice and Bob do the measurement in the following basis. 
\begin{equation}\label{eq13}
\begin{array}{l}
\ket{\xi_1} = \ket{\mathbf{0}} + \ket{\mathbf{1}} + \cdots  + \ket{\mathbf{d-1}}, \\[0.5 ex]

\ket{\xi_k} = \sum_{j=\mathbf{0}}^{\mathbf{d-1}} \left(e^{\frac{2\pi i}{d}}\right)^{k^\prime j} \ket{j}.
\end{array}
\end{equation}
where $k = 2,3,\dots, d$, $k^\prime = k-1$ and $i=\sqrt{-1}$. The definitions of $\mathbf{0}$, $\mathbf{1}$, \dots, $\mathbf{d-1}$ depend on the measurement outcomes of Alice and Bob when they do the measurements defined by the operators of (\ref{eq12}). In this way, the protocol for entanglement-assisted discrimination can be completed.
\end{proof}

\begin{remark}\label{rem2}
It is important to notice that though the dimensions of the subsystems are increasing. The entanglement-assisted discrimination can still be done by using a minimum-dimensional entangled state. Thus, we say that this discrimination process is an efficient one and here the resource state provides an advantage irrespective of the dimensions of the subsystems. 
\end{remark}
Since, the states of any present set cannot be conclusively locally distinguishable, it is also important to ask how much entanglement will be sufficient for the discrimination of the states conclusively with some nonzero probability under LOCC. In this context, we prove the following theorem.

\begin{theorem}\label{theo2}
All pure entangled states are useful to distinguish the states of any present set conclusively under LOCC.
\end{theorem}

\begin{proof}
Any pure entangled state can be converted to a two-qubit maximally entangled state with nonzero probability under LOCC \cite{Vidal99}. If the conversion is not successful then the process may result in a separable state. If the conversion is successful then, by Theorem \ref{theo1}, it is possible to distinguish the states perfectly. Otherwise, the states can be left without distinguishing. In that case, possibility of inconclusive outcome occurs. Now, perfect discrimination with some nonzero probability means conclusive identification with some nonzero probability. In this way, it is proved that all pure entangled states are useful to distinguish the states of any present set conclusively under LOCC.
\end{proof}

\subsection*{Other constructions}
Here we mention that though the entangled states that we consider for constructing present sets, are maximally entangled, it is not necessary to stick to such entangled states for the present results. Such an example can be given as the following. First, we consider four states:
\begin{equation}\label{eq14}
\begin{array}{l}
\ket{\phi_1} = a\ket{00}+b\ket{11},~\ket{\phi_2} = b\ket{00}-a\ket{11},\\[0.5 ex]

\ket{\phi_3} = c\ket{22}+d\ket{33},~ \ket{\phi_4} = d\ket{22}-c\ket{33},\\[0.5 ex]
\end{array}
\end{equation}
where $a,b,c,d$ are real numbers such that $0<$ $a$, $b$, $c$, $d<1$, $a^2+b^2=1=c^2+d^2$. Then, we consider the following entangled states:
\begin{equation}\label{eq15}
\begin{array}{l}
\ket{\alpha_1} = e\ket{\phi_1} + f\ket{\phi_3},~ \ket{\alpha_2} = f\ket{\phi_1} - e\ket{\phi_3},\\[0.5 ex]

\ket{\alpha_3} = g\ket{\phi_2} + h\ket{\phi_4},~ \ket{\alpha_4} = h\ket{\phi_2} - g\ket{\phi_4},
\end{array}
\end{equation}
where $e,f,g,h$ are also real numbers such that $0<$ $e$, $f$, $g$, $h<1$, $e^2+f^2=1=g^2+h^2$. Next, we consider $\ket{\alpha_5} = \ket{01}$. So, the set $\{\ket{\alpha_i}\}_{i=1}^5$ is also locally indistinguishable and a two-qubit maximally entangled state is sufficient to distinguish them by LOCC perfectly.

The local indistinguishability follows from the similar fact that superposing the state $\ket{\alpha_1}$ with the product states from the orthogonal subspace, it is not possible to produce a product state. However, for entanglement-assisted discrimination, the protocol should be revised. The first measurement remains the same.
\begin{equation}\label{eq16}
\begin{array}{l}
\Pi_1 = \ketbra{00}{00}+\ketbra{11}{11}+ \ketbra{21}{21} + \ketbra{31}{31},\\[0.5 ex]

\Pi_2 = \ketbra{01}{01}+\ketbra{10}{10}+ \ketbra{20}{20} + \ketbra{30}{30}.
\end{array}
\end{equation}
This measurement is performed by both Alice and Bob. Then, $\ket{\alpha_5}$ is either identified or eliminated. Further, Alice does a measurement in the following basis:
\begin{equation}\label{eq17}
\begin{array}{l}
\ket{\mathbf{0}} + \ket{\mathbf{1}} + \ket{\mathbf{2}} + \ket{\mathbf{3}}, \\[0.5 ex]

\ket{\mathbf{0}} + \ket{\mathbf{1}} - \ket{\mathbf{2}} - \ket{\mathbf{3}}, \\[0.5 ex]

\ket{\mathbf{0}} - \ket{\mathbf{1}} - \ket{\mathbf{2}} + \ket{\mathbf{3}}, \\[0.5 ex]

\ket{\mathbf{0}} - \ket{\mathbf{1}} + \ket{\mathbf{2}} - \ket{\mathbf{3}},
\end{array}
\end{equation}
where the definitions of $\ket{\mathbf{0}}$, $\ket{\mathbf{1}}$, $\ket{\mathbf{2}}$ and $\ket{\mathbf{3}}$ again depend on the previous measurement outcome. After performing this measurement, corresponding to each outcome a set of orthogonal states is produced on Bob's side. These states can be distinguished locally.

\section{Conclusion and open problems}\label{sec4}
We have constructed small sets of $d+1$ states in all two-qudit systems; $d$ is the dimension of the smaller subsystem. Such a set contains only orthogonal pure states, $d$ of them are entangled and one is a product state. We have proved that for such a set not all states are conclusively locally identifiable. Therefore, obviously, the states of the set cannot be perfectly distinguished by LOCC. Then, we have shown that a two-qubit maximally entangled state as a resource is sufficient to distinguish the states of any present set under LOCC, irrespective of the dimensions of the subsystems. This also implies that all pure entangled states are useful in distinguishing the states of any present set conclusively (with nonzero probability) under LOCC. Furthermore, based on our results, one can also think about increasing the cardinality of the sets for understanding the role of classical communication in the discrimination process. 

For further research, one can think about working on the necessary conditions in the entanglement-assisted local state discrimination setting for the present sets.

\section*{Acknowledgment}
S.G. would like to acknowledge the support from the EPSRC-funded Quantum Advantage Pathfinder (QAP) project, grant reference EP/X026167/1.

\bibliography{ref1}

\end{document}